\documentclass[aps,pre,twocolumn,superscriptaddress,floatfix,footinbib,notitlepage,groupaddress,showpacs]{revtex4-1} 
\usepackage{times,amsfonts,amsmath,amssymb,amsthm,bm,bbm,mathrsfs,graphics,graphicx,color,xcolor,enumerate,hyperref}
\hypersetup{colorlinks,linkcolor={blue},citecolor={blue},urlcolor= {blue}}

\newcommand{\ignore}[1]{}
\newtheorem{theorem}{Theorem}

\newcommand{\ket}[1]{|#1\rangle}
\newcommand{\bra}[1]{\langle #1|}
\newcommand{\blue}[1]{\textcolor{blue}{#1}}
\let\oldsqrt\sqrt
\def\sqrt{\mathpalette\DHLhksqrt}
\def\DHLhksqrt#1#2{
\setbox0=\hbox{$#1\oldsqrt{#2\,}$}\dimen0=\ht0
\advance\dimen0-0.2\ht0
\setbox2=\hbox{\vrule height\ht0 depth -\dimen0}
{\box0\lower0.4pt\box2}}

\DeclareFontFamily{OT1}{pzc}{}
\DeclareFontShape{OT1}{pzc}{m}{it}
              {<-> s * [1.25] pzcmi7t}{}
\DeclareMathAlphabet{\mathpzc}{OT1}{pzc}
                                 {m}{it}

\begin{document}

\title{Quantum detailed balance conditions and fluctuation relations for thermalizing quantum dynamics}

\author{M. Ramezani}
\affiliation{Department of Physics, Sharif University of Technology, Tehran 14588, Iran}
\affiliation{School of Physics, Institute for Research in Fundamental Sciences (IPM), Tehran 19395, Iran}

\author{F. Benatti}
\affiliation{Department of Physics, University of Trieste, I-34151 Trieste, Italy}
\affiliation{National Institute for Nuclear Physics (INFN), Trieste Section, I-34151 Trieste, Italy}

\author{R. Floreanini}
\affiliation{National Institute for Nuclear Physics (INFN), Trieste Section, I-34151 Trieste, Italy}

\author{S. Marcantoni}
\affiliation{Department of Physics, University of Trieste, I-34151 Trieste, Italy}
\affiliation{National Institute for Nuclear Physics (INFN), Trieste Section, I-34151 Trieste, Italy}

\author{M. Golshani}
\affiliation{Department of Physics, Sharif University of Technology, Tehran 14588, Iran}
\affiliation{School of Physics, Institute for Research in Fundamental Sciences (IPM), Tehran 19395, Iran}

\author{A. T. Rezakhani}
\affiliation{Department of Physics, Sharif University of Technology, Tehran 14588, Iran}

\date{\today}

\begin{abstract}
Quantum detailed balance conditions and quantum fluctuation relations are two important concepts in the dynamics of open quantum systems: both concern how such systems behave when they thermalize because of interaction with an environment. We prove that for thermalizing quantum dynamics the quantum detailed balance conditions yield validity of a quantum fluctuation relation (where only forward-time dynamics is considered). This implies that to have such a quantum fluctuation relation (which in turn enables a precise formulation of the second law of thermodynamics for quantum systems) it suffices to fulfill the quantum detailed balance conditions. We, however, show that the converse is not necessarily true; indeed, there are cases of thermalizing dynamics which feature the quantum fluctuation relation without satisfying detailed balance. We illustrate our results with three examples.
\end{abstract}
\maketitle

\section{Introduction}
\label{sec:intro}

Thermodynamics is a successful theory to describe (equilibrium) properties of macroscopic open systems \cite{book:Callen}. Among the three laws of thermodynamics, the second law has a fundamental and distinct feature. This law governs how open systems interacting with their ambient environment tend to equilibrate or thermalize with the environment, and in this sense it naturally incorporates the concept of \textit{irreversibility}. This peculiar feature of the second law begs the question of how one can explain its emergence from fundamental laws of nature.  

The ``detailed balance condition" \cite{book:Callen}, roughly stating that at equilibrium each elementary process (formally, ``$i$"$\to$``$j$") and its reverse (``$j$"$\to$``$i$") need to be equally probable ($p^{(\mathrm{eq})}_{i}w_{i\to j}=p^{(\mathrm{eq})}_{j}w_{j\to i}$, where $p^{(\mathrm{eq})}_{i}$ and $w$ denote, respectively, the probability of a state at equilibrium and the state transition rates), has gained a pivotal role in understanding the process of equilibration (and thermalization). It was employed by Boltzmann in proving his $H$-theorem and by Maxwell in the development of kinetic theory. The validity of this condition has been attributed to fundamental symmetries of the basic dynamical laws of nature under \textit{time-reversal}---or to the very concept of microscopic \textit{reversibility} \cite{Jarzynski-annualreview}.  

Another relevant and ubiquitous feature in the behavior of open systems described by thermodynamics at equilibrium is embodied by the \textit{fluctuations} of their properties  around average values given by thermodynamics---which typically diminish when the open system becomes large. A powerful approach to study fluctuations of thermodynamic quantities is provided by ``fluctuation theorems" \cite{Evans-Searles}. The interest in this subject has been specially spurred and reinvigorated recently by the derivation of interesting and important ``fluctuation relations" by Jarzynski \cite{1997-Jarzynski} and Crooks \cite{1999-Crooks}. Such relations connect the work done on a classical system by an external driving force to the equilibrium free energy difference between the initial and the final states of the system. The Crooks relation, in particular, which is the more general of the two and implies the Jarzynski equality as a corollary, compares the probability of \textit{doing} a certain amount of work under a driving protocol in the \textit{forward}-time direction with the probability of \textit{extracting} the same amount of work in the \textit{backward}-time (i.e., time-reversed) protocol, providing a refined statement of the second law of thermodynamics. Other similar relations have also been studied since then, also for thermodynamic quantities other than work \cite{2012-Seifert,book:Evans}.

Given the fundamentally different features of \textit{quantum} mechanics and the diversity of quantum dynamics in contrast to classical mechanics and dynamics, the situation with either the detailed balance conditions and the fluctuation relations becomes even more interesting for quantum systems. In fact, both concepts have been extended to the quantum domain and extensively studied in various aspects---see, e.g., Refs. \cite{1972-Kossakowski,1976-Alicki,1977-Kossakowski,1984-Majewski,1998-Majewski,2007-Fagnola,2008-Fagnola,2010a-Fagnola,2010b-Fagnola,thesis,Db-ent,2000-Kurchan,2000-Tasaki,2009-Esposito,2011-Campisi,2009-Campisi,2004-Jarzynski,2013-Albash,2014-Rastegin,2015-Goold,2015-Aurell,2015-Manzano,2014-Jaksic,2018-Ramezani,2008-Crooks,Crooks-T} (and the references therein). These subjects constitute part of the emerging field of quantum thermodynamics \cite{qthermo,qthermo-2,2016-Alipour}.  

Considering that time-reversal is in the heart of both fluctuation relations and detailed balanced conditions, it seems natural that these concepts should be intimately related. Indeed, in the proof of the Crooks fluctuation relation the detailed balance condition has been used \cite{1999-Crooks,2000-Crooks}. Yet, and up to the best of our knowledge, a systematic and comprehensive investigation of this relation and of the implications of either concept on the other one is still lacking for quantum systems. 

In this paper we partially bridge this gap for the more general case of \textit{open}-system quantum dynamics which \textit{thermalize}. In particular, we rigorously prove that the quantum detailed balance (QDB) conditions imply (a forward-forward version of) the quantum fluctuation relations (QFRs), but the converse is not necessarily valid. In doing so, we first extend a recently proposed forward-forward version of the QFRs for \textit{heat} exchange in open quantum systems \cite{2018-Ramezani}, which in contrast to forward-backward Crooks-like QFRs deals with the ratio of two probabilities along the forward path. We discuss some conditions for thermalizing dynamics under which this relation holds. Next we consider two versions of the QDB condition and show how they can enable this QFR. We supply examples to illustrate our results, and in particular to underline the point that to have the QFR we do not necessarily need the QDB condition---that is, QDB $\Rightarrow$ QFR but QFR $\not\Rightarrow$ QDB.   

This paper is structured as follows. In Sec. \ref{sec: Therm}, we review some basics about the dynamics of open quantum systems and define the class of dynamics we are interested in, namely those dynamics with a unique asymptotic state which is thermal. In Sec. \ref{sec: FluctRel}, we prove two results regarding the validity of the QFR: (i) an asymptotic QFR for energy exchange in an open quantum system undergoing a thermalizing dynamics, and (ii) for finite time thermalizing dynamics of qubits whose fixed-point states is also thermal. Section \ref{sec: DetBal}  is devoted to the description of two important QDB conditions introduced in the literature. This property will be related to the QFR in Sec. \ref{sec: FluctDet}. We illustrate our results with three different examples in Sec. \ref{sec: Examples}. The results are summarized in Sec. \ref{sec:conc}.

\section{Thermalizing dynamics}
\label{sec: Therm}

The dynamics of an open quantum system with the Hilbert space $\mathpzc{H}\equiv \mathbb{C}^{d}$, uncoupled with an \textit{environment} initially at time $\tau=0$, but later interacting with its environment is described by a (one-parameter family of) linear completely-positive, trace preserving (CPTP) quantum \textit{dynamical maps} or \textit{channels}, which transform any initial density matrix $\varrho(0)$ of the system into \cite{book:Nielsen}
\begin{equation}
\label{dymaps}
\varrho(\tau)=\mathpzc{G}_{\tau}[\varrho(0)],\, \tau\geqslant 0.
\end{equation}
To these dynamical maps one can also associate \textit{dual} dynamical maps $\mathpzc{G}^{\sharp}_{\tau}$ in the Heisenberg picture defined through 
\begin{equation}
\mathrm{Tr}\big[\mathpzc{G}_{\tau}[\sigma]A\big]=\mathrm{Tr}\big[\sigma \mathpzc{G}^{\sharp}_{\tau}[A]\big],
\label{eq:Shro-Heis}	
\end{equation}
for all density matrices $\sigma\in\mathpzc{H}$ and all (bounded) $d\times d$ complex matrices $A$ ($\in M_d(\mathbb{C})$.
 
Let us also assume that the (bare) open quantum system in absence of the environment is described by a time-independent Hamiltonian
\begin{equation}
\label{Ham}
H=\sum_{m=1}^{d}E_{m}\ket{m}\bra{m},\, (E_{m}\in\mathbb{R},\,\langle m|m'\rangle=\delta_{mm'}),
\end{equation}
and let us associate with that the following equilibrium or thermal state at inverse temperature $\beta$:
\begin{equation}
\label{therstates}
\varrho^{(\beta)} =e^{-\beta H}/\mathrm{Tr}[e^{-\beta H}].
\end{equation}
It is known that any CPTP dynamical map (equivalently called quantum ``channel" or ``operation") can always be written through the Kraus representation
\begin{equation}
\label{channel}
\mathpzc{G}_{\tau}[\cdot]=\sum_{j}G^{(j)}_{\tau}\,\cdot\, G^{(j)\dag}_{\tau},
\end{equation}
where $G^{(j)}_{\tau}\in M_d(\mathbb{C})$ and $\sum_j G^{(j)\dag}_{\tau} G^{(j)}_{\tau}=\mathbbmss{I}$ \cite{book:Nielsen}. However, under some specific conditions such as weak coupling with the environment, the Born-Markov approximation, and the secular approximation, one can show that the dynamics of the system can be recast through the \textit{master equation} \cite{2002-Breuer,hbar}
\begin{equation}
\label{ME}
\partial_{\tau}\varrho(\tau)=\mathpzc{L}[\varrho(\tau)],
\end{equation}
where
\begin{equation}
\mathpzc{L}[\cdot]=-i[H,\cdot] + \sum_{k,l=1}^{d^2-1} C_{kl}\big(F_k \cdot F^{\dag}_l - (1/2)\big\{F_l^{\dag} F_k , \cdot\big\}\big)
\label{Lindblad1}
\end{equation}
is the time-independent generator of the dynamics in the Lindblad form, in the sense that $\partial_{\tau}\mathpzc{G}_{\tau}=\mathpzc{L}\circ\mathpzc{G}_{\tau}$ or $\mathpzc{G}_{\tau}=e^{\tau\mathpzc{L}}$, $C=[C_{kl}]$ is a positive semidefinite matrix, $\{F_{k}\}_{k=1}^{d^{2}}$ is a set of suitable orthonormal basis matrices such that $F_{d^2}=\mathbbmss{I}/d$ and $\mathrm{Tr}[F^{\dag}_{k}F_l]=\delta_{kl}$, and the map composition $\circ$ is understood as $\mathpzc{G}_{s}\circ\mathpzc{G}_{t}[\cdot]=\mathpzc{G}_{s}\big[\mathpzc{G}_{t}[\cdot]\big]$. In this case $\mathpzc{G}_{\tau}$ would satisfy the \textit{semigroup} composition law
\begin{equation}
\label{semglaw}
\mathpzc{G}_{\tau_1+\tau_2}=\mathpzc{G}_{\tau_1}\circ\mathpzc{G}_{\tau_2},\,\forall  \tau_{1,2}\geqslant 0.
\end{equation}

It will be useful later in the paper to note that if the dynamics is in the form of Eq. (\ref{ME}) with a Lindblad generator as in Eq. (\ref{Lindblad1}) (i.e., $\mathpzc{G}_{\tau}=e^{\tau \mathpzc{L}}$), the dual dynamical map $\mathpzc{G}^{\sharp}_{\tau}=e^{\tau \mathpzc{L}^{\sharp}}$ on observables is given by the dual of the generator
\begin{equation}
\label{Lindblad1.2}
\mathpzc{L}^{\sharp}[\cdot]=i[H,\cdot]+\sum_{k,l=1}^{d^2-1}C_{kl}\big(F^{\dag}_l \cdot F_k- (1/2)\big\{F_l^{\dag} F_k, \cdot\big\}\big).
\end{equation}

Now that we have explained what the state of an open system can be and how to its observables one can associate a dynamics too, we end this section by stating two definitions which are pivotal in this paper.
\medskip

\textbf{Definition:} We call a dynamical map $\mathpzc{G}_{\tau}$ \textit{thermalizing} to inverse temperature $\beta$ if for any initial state $\varrho(0)$ we have
\begin{equation}
\label{Gibbs3}
\varrho(\infty)=\lim_{\tau\to\infty}\mathpzc{G}_{\tau}[\varrho(0)]=\varrho^{(\beta)},
\end{equation} 
that is, $\varrho^{(\beta)}$ becomes the \textit{asymptotic state} of the dynamics for any initial state. 

In particular, we shall be interested in the scenario where the system initial state is thermal at inverse temperature $\beta_{\mathrm{i}}$, $\varrho(0)=\varrho^{(\beta_{\mathrm{i}})}$, and the final one is thermal at inverse temperature $\beta_{\mathrm{f}}$, $\varrho(\infty)=\varrho^{(\beta_{\mathrm{f}})}$ (with $\beta_{\mathrm{f}}\neq\beta_{\mathrm{i}}$). 

As we shall see later, thermalizing properties of the system dynamics need not require that the environment be a heat bath in equilibrium at inverse temperature $\beta_{\mathrm{f}}$. Furthermore, the dynamical map $\mathpzc{G}_{\tau}$ may not obey a semigroup composition law and may show memory and non-Markovian effects. In the latter case, although $\varrho^{(\beta_{\mathrm{f}})}$ is the asymptotic state for $\mathpzc{G}_{\tau}$, in general it is not $\mathpzc{G}_{\tau}$-invariant.
\medskip

\textbf{Definition:} A dynamical map is called \textit{fixed-point thermalizing} (FPT) if it is thermalizing (to some inverse temperature $\beta_{\mathrm{f}}$) and the thermal state is its fixed point too, 
\begin{equation}
\mathpzc{G}_{\tau}[\varrho^{(\beta_{\mathrm{f}})}]=\varrho^{(\beta_{\mathrm{f}})},\,\forall \tau\geqslant0.
\label{fix-point-maps}
\end{equation} 

\section{Quantum fluctuation relation (QFR)}
\label{sec: FluctRel}

The Crooks QFR for a dynamical system is a typical instance of fluctuation relations. It states that   
\begin{equation}
\frac{P_{\textsc{f}}(+\mathbbmss{W};\tau)}{P_{\textsc{r}}(-\mathbbmss{W};\tau)} =e^{\beta(\mathbbmss{W}-\Delta \mathbbmss{F})},
\end{equation}
that is, the ratio of the probability of \textit{doing} a certain amount of work $\mathbbmss{W}$ under a driving protocol in the \textit{forward}-time direction (hence the subscript ``\textsc{f}") and the probability of \textit{extracting} the same amount of work in the \textit{backward}-time (``\textsc{r}") protocol is determined by the difference between the initial and final free energies. This (forward-backward) relation has been extended in numerous respects both for classical and quantum systems and also for quantities other than work.

More recently, however, a distinct QFR has been proposed in Ref. \cite{2018-Ramezani} for the case of irreversible dynamics that cannot be run backward in time. In such a context, unlike the typical case of forward-backward Crooks-like QFRs, both probabilities are calculated along the forward-time direction (hence ``forward-forward"). Specifically, this QFR concerns heat exchange in an open quantum system evolving in time through a thermalizing Markovian dynamics, and shows that the probability $P(+\mathbbmss{Q};\tau)$ of absorbing a certain amount of heat $\mathbbmss{Q}$ from the environment at time $\tau$ is related to the probability of releasing the same amount of heat $P(-\mathbbmss{Q};\tau)$ via an exponential factor, which depends on $\mathbbmss{Q}$ and on the difference between the initial inverse temperature of the system, $\beta_{\mathrm{i}}$, and the asymptotic temperature determined by the dynamics, $\beta_{\mathrm{f}}$,
\begin{equation}
\label{fluct0}
\frac{P(+\mathbbmss{Q};\tau)}{P(-\mathbbmss{Q};\tau)}=e^{\Delta\beta\,\mathbbmss{Q}}, \quad 
\Delta\beta:=\beta_{\mathrm{i}}-\beta_{\mathrm{f}},
\end{equation}
where we have removed the subscript ``\textsc{f}." An interesting feature of the above expression is that although both $P(+\mathbbmss{Q};\tau)$ and $P(-\mathbbmss{Q};\tau)$ are time-dependent, their ratio is time-independent. Comparing two probabilities both evaluated along the \textit{forward}-time path removes the issue of defining the reverse path for an irreversible dynamics. Moreover, such a result does not require unitary (closed-system) dynamics, thus in these two specific senses this approach may be complementary to large part of the existing literature on QFRs.

Here we provide a generalized framework for the forward-forward QFR. Given the setting of the previous section, the probability that the $d$-level system (prepared initially at the thermal state $\varrho(0)= \varrho^{(\beta_{\mathrm{i}})}$) \textit{absorbs} a positive amount of energy $\mathbbmss{E}\geqslant 0$ from the environment in the time interval $\tau$ is given by
\begin{equation}
P(+\mathbbmss{E};\tau)=\sum_{mn}p_{m}(\beta_{\mathrm{i}})\,p(\ket{m}\rightarrow\ket{n};\tau)\,\delta(\mathbbmss{E}-(E_{n}-E_{m})),
\label{heat-absorption-probability}
\end{equation}   
where $p(\ket{m}\rightarrow\ket{n};\tau)$ is the transition probability from the state $\ket{m}$ to the state $\ket{n}$ in the time interval $\tau$,
\begin{equation}
p(\ket{m}\rightarrow\ket{n};\tau)=\bra{n}\mathpzc{G}_{\tau}[\ket{m}\bra{m}]\ket{n}.
\label{eq:transition-probability}
\end{equation}
and $p_m(\beta_{\mathrm{i}})$ is the probability that the system is found in the state $\ket{m}\bra{m}$,
\begin{equation}
p_{m}(\beta_{\mathrm{i}})=\bra{m}\varrho^{(\beta_{\mathrm{i}})} \ket{m}=e^{-\beta_{\mathrm{i}} E_{m}}/\mathrm{Tr}[e^{-\beta_{\mathrm{i}} H}].
\label{initial-probability}
\end{equation}
In terms of the Kraus operators in Eq. \eqref{channel} whose entries are $[G^{(j)}_{\tau}]_{nm}=\langle n\vert G^{(j)}_{\tau}\vert m\rangle$ (with respect to the Hamiltonian eigenbasis $\{|m\rangle\}_{m=1}^{d}$), we can rewrite
\begin{equation}
p(\ket{m}\rightarrow\ket{n};\tau)=\sum_{j}|[G^{(j)}_{\tau}]_{nm}|^2.
\label{transition-probability-kraus}
\end{equation}
Likewise, the probability that the system \textit{releases} the amount of energy $\mathbbmss{E}\geqslant 0$ to the environment in the time interval $\tau$ is given by
\begin{equation}
P(-\mathbbmss{E};\tau)=\sum_{mn}p_{n}(\beta_{\mathrm{i}})\,p(\ket{n}\rightarrow\ket{m};\tau)\,\delta(\mathbbmss{E}-(E_{n}-E_{m})).
\label{heat-release-probability}
\end{equation}

Rather than comparing two probability distributions related to different dynamics, e.g., corresponding to a ``forward" protocol and a ``backward" protocol (as usually done in the literature about QFRs \cite{2011-Campisi}), following the approach presented recently in Ref. \cite{2018-Ramezani}, we concentrate on the ratio $P(+\mathbbmss{E};\tau)/P(-\mathbbmss{E};\tau)$. In particular, in the following sections we shall provide instances of dissipative QFR of the form [cf. Eq. \eqref{fluct0}]
\begin{equation}
\label{FluctRel}
R(\mathbbmss{E};\tau)\equiv\frac{P(+\mathbbmss{E};\tau)}{P(-\mathbbmss{E};\tau)}=e^{\Delta\beta\,\mathbbmss{E}},
\end{equation}
where $\Delta\beta=\beta_{\mathrm{i}}-\beta_{\mathrm{f}}$, with $\beta_{\mathrm{f}}$ being the asymptotic inverse temperature reached by the thermalization process. 

Note that,  if this relation holds, when $\beta_{\mathrm{i}}<\beta_{\mathrm{f}}$, that is, when the initial temperature is larger than the final one, the probability of absorbing a certain amount of energy $\mathbbmss{E}>0$ by the system is exponentially smaller than the probability of releasing the same amount of energy to the environment. If we assume that the thermalization process is due to the interaction with a thermal environment at inverse temperature $\beta_{\mathrm{f}}$, and that there is no ``work" contribution to the exchange of energy, Eq. \eqref{FluctRel} constitutes a precise mathematical statement for the observation that heat is expected to flow from the hot body to the cold one---in accordance with the Clausius statement of the second law of thermodynamics \cite{book:Callen}.

We remark that if the system dynamics is generated by a Lindblad-like \textit{time-dependent} generator $\mathpzc{L}_{\tau}$ whose form is akin to the form (\ref{Lindblad1}) but with a time-dependent Hamiltonian $H_{\tau}$ and a time-dependent matrix $C(\tau)$, then the total energy exchange rate between the system and the environment at time $\tau$ amounts to
\begin{equation}
\label{Alickitherm}
\partial_{\tau}\mathrm{Tr}\big[\varrho(\tau)\,H_{\tau}\big] =\mathrm{Tr}\left[\partial_{\tau}\varrho(\tau)\,H_{\tau}\right]+ \mathrm{Tr}\left[\varrho(\tau)\partial_{\tau}H_{\tau}\right] .
\end{equation}
The first term describes the \textit{work} exchange rate and the second the \textit{heat} exchange rate; only the latter contributes if the system Hamiltonian $H_{\tau}$
is not explicitly time-dependent, whence the energy $\mathbbmss{E}$ absorbed or released by the system can be interpreted as exchanged heat. 

Before investigating the connection between the QFR (\ref{FluctRel}) and the QDB conditions, we prove two results about the QFR. We need to point out to two useful relations for the transition probabilities:

(i) It is immediate to see that when 
\begin{equation}
\label{aux1}
e^{-\beta_{\mathrm{f}}E_m}\,p(\ket{m}\rightarrow\ket{n};\tau)=e^{-\beta_{\mathrm{f}}E_n}\,p(\ket{n}\rightarrow\ket{m};\tau),
\end{equation}
the QFR (\ref{FluctRel}) evidently holds. 

(ii) For any FPT dynamical map we have 
\begin{equation}
\sum_{n}p_{n}(\beta_{\mathrm{f}})\,p(\ket{n}\rightarrow\ket{m};\tau)=p_{m}(\beta_{\mathrm{f}}).
\label{prop-}
\end{equation}
This can be verified as
\begin{align*}
p_{m}(\beta_{\mathrm{f}})\overset{(\ref{fix-point-maps})}{=}\big[\mathpzc{G}_{\tau}[\varrho^{(\beta_{\mathrm{f}})}]\big]_{mm}\overset{(\ref{eq:transition-probability})}{=}\sum_{n}p(\ket{n}\rightarrow\ket{m};\tau)\,p_{n}(\beta_{\mathrm{f}}).
\end{align*}

Now we show that although arbitrary thermalizing dynamical maps do not necessarily satisfy the QFR (\ref{FluctRel}) \textit{instantaneously}, they all fulfill this property for \textit{asymptotically long times}.
\begin{theorem}
\label{theo: infinity}
For any thermalizing dynamical map $\mathpzc{G}_{\tau}$, 
\begin{equation}
R(\mathbbmss{E},\infty)=e^{\Delta\beta\,\mathbbmss{E}}.
\label{FR-infty}
\end{equation}
\end{theorem}

\begin{proof}
We have
\begin{equation}
\label{thermprob}
p(\ket{m}\rightarrow\ket{n};\infty)\overset{(\ref{eq:transition-probability})}{=}\bra{n}\mathpzc{G}_{\infty}[\ket{m}\bra{m}]\ket{n}\overset{(\ref{Gibbs3})}{=}\bra{n}\varrho^{(\beta_{\mathrm{f}})}\ket{n}. 
\end{equation}
Thus Eq. \eqref{aux1} holds, which in turn yields relation (\ref{FluctRel}).
\end{proof}

For the case of FPT dynamical maps, no advantages follow compared to Theorem \ref{theo: infinity} except for qubits ($d=2$), where we have the following theorem. 

\begin{theorem}
\label{theo: FP-thermalizing}
Every FPT dynamical map for qubits satisfies the QFR \eqref{FluctRel}.
\end{theorem}

\begin{proof}
We have
\begin{equation*}
p_{1}({\beta_{\mathrm{f}}})\overset{(\ref{prop-})}{=}p(\ket{1}\rightarrow\ket{1};\tau)\,p_{1}(\beta_{\mathrm{f}})+p(\ket{2}\rightarrow\ket{1};\tau)\,p_{2}(\beta_{\mathrm{f}}).
\end{equation*}
However, in the case of a qubit, $p(\ket{1}\rightarrow\ket{1};\tau)=1-p(\ket{1}\rightarrow\ket{2};\tau)$; whereby the above equation fulfills Eq. (\ref{aux1}).
\end{proof}

Before progressing further, it is helpful to make some remarks and summarize our findings thus far. (i) The QFR \eqref{FluctRel} can also be formulated for discrete-time dynamics, where both of the above theorems will still apply.  (ii) In contrast to Ref. \cite{2018-Ramezani}, here we have not restricted to the case of time-independent Hamiltonians. Since in such cases one cannot unambiguously associate the variation of energy in the open system solely to heat, we have always referred to ``energy exchange" (rather than heat exchange). (iii) Although Ref. \cite{2018-Ramezani} assumed a Markovian Lindblad form for the thermalizing dynamics, thus far we have not assumed any particular type of open-system dynamics. Yet, we have shown that the dynamics suffices to be thermalizing to enforce the QFR at least \textit{asymptotically}; and further, if it has the extra property that its asymptotic thermal state is also its fixed point this will guarantee the \textit{finite-time} QFR at least for qubits. Although Theorems \ref{theo: infinity} and \ref{theo: FP-thermalizing} show some cases of thermalizing dynamics where the QFR (\ref{FluctRel}) holds, it still remains to find general sufficient conditions for an open-system dynamics to fulfill the QFR for finite times. This is exactly where we employ the QBD condition. 

\section{Quantum detailed balance (QBD)}
\label{sec: DetBal}

Among numerous existing extensions of the classical detailed balance conditions to quantum systems, we shall follow the general approach proposed in Refs. \cite{2007-Fagnola,2008-Fagnola,2010a-Fagnola,2010b-Fagnola} for their generality (see Ref. \cite{thesis} for a review). This is based on turning the algebra of observables $M_d(\mathbb{C})$ into a $d^2$-dimensional Hilbert space $\mathpzc{H}_{\,\Sigma,s}$ by means of the scalar product
\begin{equation}
\langle\hskip-.6mm\langle A,B \rangle\hskip-.6mm\rangle_{s}=\mathrm{Tr}\big[\Sigma^{1-s}A^{\dag}\Sigma^s B\big], \, A,B\in M_d(\mathbb{C}),
\label{eq:scalar-product}
\end{equation}
where $s\in[0,1]$ and $\Sigma$ is a given full-rank \textit{reference state} (i.e., $\Sigma>0$). This scalar product makes the matrix algebra $M_d(\mathbb{C})$ a Hilbert space $\mathpzc{H}_{\,\Sigma,s}$ (isomorphic to $\mathbb{C}^{d^2}$). In addition, given a linear map $\mathpzc{O}$ on $\mathpzc{H}_{\,\Sigma,s}$, one can define its \textit{adjoint} $\mathpzc{O}^{\star}$ relative to this scalar product by
\begin{equation}
\label{adjoint}
\langle\hskip-.6mm\langle A,\mathpzc{O}[B]\rangle\hskip-.6mm\rangle_s=\langle\hskip-.6mm\langle \mathpzc{O}^{\star}[A],B\rangle\hskip-.6mm\rangle_s.
\end{equation}
Note the difference in notation between the adjoint operation ``$\dag$" with respect to the standard scalar product $\langle\, ,\rangle$ on $\mathbb{C}^d$ and the adjoint operation ``$\star$" with respect to the scalar product $\langle\hskip-.6mm\langle\,,\rangle\hskip-.6mm\rangle_s$ on $M_d(\mathbb{C})$.

The first QDB condition refers to a dynamical \textit{semigroup} map $\mathpzc{G}_{\tau}=e^{\tau\mathpzc{L}}$ with the generator in the Lindblad form (\ref{Lindblad1}) and its dual map $\mathpzc{G}^{\sharp}_{\tau}=e^{\tau \mathpzc{L}^{\sharp}}$ with the generator (\ref{Lindblad1.2}).
\medskip

\textbf{Definition 1.} Let $\mathpzc{L}^{\sharp\star}$ be the adjoint of $\mathpzc{L}^{\sharp}$ in Eq. \eqref{Lindblad1.2} with respect to the scalar product \eqref{eq:scalar-product}. We say a dynamical map $\mathpzc{G}_{\tau}=e^{\tau \mathpzc{L}}$ (or equivalently $\mathpzc{G}^{\sharp}_{\tau}=e^{\tau \mathpzc{L}^{\sharp}}$) has the QDB property with respect to a reference state $\Sigma>0$ if
\begin{equation}
\label{QDB1}
\mathpzc{L}^{\sharp}[A]-\mathpzc{L}^{\sharp\star}[A]=2i[H,A],\,\forall A\in M_d(\mathbb{C}).
\end{equation}

It is straightforward to see that such a requirement is satisfied (with a vanishing right-hand side) in the case of the generator of a classical Pauli equation---of which Eq. \eqref{QDB1} is a quantum generalization accounting for the presence of a contribution coming from the commutator with a Hamiltonian $H$. In addition, an immediate consequence of this condition (if \blue{it} holds) is that the reference state must be $\mathpzc{G}_{\tau}$-invariant. This can be seen as follows. Equation \eqref{Lindblad1.2} implies $\mathpzc{L}^{\sharp}[\mathbbmss{I}]=0$, which in turn combined with the above QDB condition gives $\mathpzc{L}^{\sharp\star}[\mathbbmss{I}]=0$. Now, if we replace $A=\mathbbmss{I}$ and $\mathpzc{O}=\mathpzc{L}^{\sharp}$ in Eq. \eqref{adjoint}, we obtain
\begin{equation*}
\mathrm{Tr}\big[\Sigma\,\mathpzc{L}^{\sharp}[B]\big]=\mathrm{Tr}\big[\Sigma^{1-s}\mathpzc{L}^{\sharp\star}[\mathbbmss{I}]\,\Sigma^s B\big]=0,\,\forall B\in M_d(\mathbb{C}).
\end{equation*}
This yields that $\mathpzc{L}[\Sigma]=0$, that is, $\mathpzc{G}_{\tau}[\Sigma]=\Sigma$.

Since not all dynamical maps have the semigroup property, it is important to consider a second QDB condition which does not refer to the semigroup properties of a dynamical map $\mathpzc{G}_{\tau}$, but only to its behavior with respect to the \textit{time-reversal operation} $\mathpzc{T}$ defined as
\begin{equation}
\label{Trevaut}
\mathpzc{T}[A]=\Theta A^{\dag}\Theta^{\dag},\,\forall A\in M_{d}(\mathbb{C}),
\end{equation}
where $\Theta$ is the \textit{time-reversal operator} \cite{book:Sakurai}. We give a brief review of the definitions and properties of these operations in appendix \ref{app:Theta}.
\medskip

\textbf{Definition 2.} A dynamical map $\mathpzc{G}^{\sharp}_{\tau}$ (in the Heisenberg picture) is said to have the QDB property with respect to a reference state $\Sigma$ if
\begin{equation}
\label{QDB2}
\langle\hskip-.6mm\langle A^{\dag},\mathpzc{G}^{\sharp}_{\tau}[B]\rangle\hskip-.6mm\rangle_s=\langle\hskip-.6mm\langle\mathpzc{T}[B^{\dag}],\mathpzc{G}^{\sharp}_{\tau}\big[\mathpzc{T}[A]\big]\rangle\hskip-.6mm\rangle_s,\,\forall A,B\in M_d(\mathbb{C}).
\end{equation}
Such a condition is based on the principle of \textit{microreversibility} that links the equilibrium probabilities of forward-time processes with those of their backward-time or time-reversed images. 

In summary, the first QDB condition \eqref{QDB1} relies on the semigroup structure of the dynamics and imposes a constraint on its generator; whereas  the second condition \eqref{QDB2} concerns general dynamical maps (independently of any composition law possibly holding among them) but employing a time-reversal linear map. It is interesting to see that how the second definition (\ref{QDB2}) compares with Eq. (\ref{QDB1}) when the dynamics is governed by a CPTP \textit{semigroup} map generated by $\mathpzc{L}^{\sharp}$. It has been shown that the condition \eqref{QDB2} matches the condition \eqref{QDB1} if the dissipative part of the generator, i.e., $\mathpzc{L}^{\sharp}[\cdot] - i[H,\cdot]$, and the Hamiltonian are both invariant under time-reversal \cite{2008-Fagnola}. Nevertheless, note that inserting $A=\mathbbmss{I}$ into Eq. \eqref{QDB2} yields $\mathrm{Tr}\big[\Sigma\,\mathpzc{G}^{\sharp}_{\tau}[B]\big]=\mathrm{Tr}\big[\Sigma\,\mathpzc{T}[B]\big]$ which---unlike the comment after Definition 1---does not imply $\mathpzc{G}_{\tau}$-invariance of the reference state $\Sigma$.

\section{QFR and QDB}
\label{sec: FluctDet}

In this section we study the relations between the two QDB conditions of Sec. \ref{sec: DetBal} and the QFR \eqref{FluctRel} and prove our main results in two theorems. We shall first consider the case of dynamical semigroup maps $\mathpzc{G}_{\tau}=e^{\tau\mathpzc{L}}$ and next the case of generic CPTP maps $\mathpzc{G}_{\tau}$. The reference state in the QDB conditions will be chosen to be the asymptotic thermal state, i.e., $\Sigma = \varrho^{(\beta_\mathrm{f})}$.

\subsection{QFR for dynamical semigroups}
\label{subsec:fr-d}

Let us decompose $\mathpzc{L}^{\sharp}$ as $\mathpzc{L}^{\sharp}=\mathpzc{L}^{\sharp}_{H}+\mathpzc{L}^{\sharp}_{D}$, where 
\begin{align}
\label{decL}
\mathpzc{L}^{\sharp}_{H}=&\frac{1}{2}(\mathpzc{L}^{\sharp}-\mathpzc{L}^{\sharp\star})\overset{(\ref{Lindblad1.2})}{=}i[H,\,],\\
\mathpzc{L}^{\sharp}_{D}=&\frac{1}{2}(\mathpzc{L}^{\sharp}+\mathpzc{L}^{\sharp\star}),
\end{align}
thus 
$\mathpzc{L}_{H}^{\sharp\star}=-\mathpzc{L}^{\sharp}_{H}$ and $\mathpzc{L}_{D}^{\sharp\star}=\mathpzc{L}^{\sharp}_{D}$.

Before stating our main results (Theorems \ref{TheoQDB1} and \ref{TheoQDB2}), we first prove some useful results.

\textbf{Lemma.} Let $\mathpzc{K}$ be a linear map on $\mathpzc{H}_{\,\Sigma,s}$. 
\begin{enumerate}
\item If $\mathpzc{K}$ is self-adjoint with respect to the scalar product \eqref{adjoint} and $(\mathpzc{K}[A])^{\dag}=\mathpzc{K}[A^{\dag}]$ ($\forall A\in M_{d}(\mathbb{C})$), then 
\begin{equation}
\label{eq:lemma-self-adjoint}
e^{-\beta_{\mathrm{f}}E_{m}}\langle m|\mathpzc{K}[|n\rangle\langle n|]|m\rangle=e^{-\beta_{\mathrm{f}}E_{n}}\langle n|\mathpzc{K}[|m\rangle\langle m|]|n\rangle.
\end{equation} 

\item If $(\mathpzc{K}[A])^{\dag}=\mathpzc{K}[A^{\dag}]$ and $(\mathpzc{K}^{\star}[A])^{\dag}=\mathpzc{K}^{\star}[A^{\dag}]$ ($\forall A\in M_{d}(\mathbb{C})$), then for the linear map $\displaystyle\mathpzc{R}_{\,s}[\cdot]=\Sigma^{1-2s}\cdot\Sigma^{2s-1}$ from $\mathpzc{H}_{\,\Sigma,s}$ into itself satisfies $\mathpzc{K}\circ \mathpzc{R}_{\,s}=\mathpzc{R}_{\,s}\circ \mathpzc{K}$.

\item Let $\mathpzc{G}^{\sharp}_{\tau}=e^{\tau \mathpzc{L}^{\sharp}}$ be a dynamical semigroup map with self-adjoint Lindblad generator $\mathpzc{L}^{\sharp}$. Let $\Sigma$ be a full-rank state with eigenvectors $\{\ket{m}\}_{m=1}^{d}$, $\mathpzc{H}_{\,\Sigma,s}$ the Hilbert space with the scalar product \eqref{eq:scalar-product}. Then the $\mathpzc{R}_{\,s}$-invariant subspace generated by the operators $\{\ket{m}\bra{m}\}_{m=1}^{d}$ and its complement (the subspace orthogonal to the former) are left-invariant by $\mathpzc{G}^{\sharp}_{\tau}$.

\end{enumerate}

\begin{proof}
\label{lemma:self-adjoint}

1. Replace $A=\ket{m}\bra{m}$ and $B=\ket{n}\bra{n}$ in Eq. \eqref{adjoint}, where $H\ket{m}=E_m\ket{m}$ and $\Sigma=\varrho^{(\beta_{\mathrm{f}})}$. 

2. We adapt the argument presented in Ref. \cite{1976-Alicki}. For arbitrary $A,B\in\mathpzc{H}_{\,\Sigma,s}$ we have
\begin{widetext}
\begin{align*}
\langle\hskip-.6mm\langle \mathpzc{K}\circ\mathpzc{R}_{\,s}[A],B\rangle\hskip-.6mm\rangle_{s}=&\langle\hskip-.6mm\langle\mathpzc{R}_{\,s}[A],\mathpzc{K}^{\star}[B]\rangle\hskip-.6mm\rangle_{s}=\mathrm{Tr}\big[\Sigma^{1-s}(\Sigma^{1-2s} A \Sigma^{2s-1})^{\dag}\Sigma^{s} \mathpzc{K}^{\star}[B]\big]=\mathrm{Tr}\big[\Sigma^{s} A^{\dag} \Sigma^{1-s}\mathpzc{K}^{\star}[B]\big]\\
=& \mathrm{Tr}\big[\Sigma^{1-s}(\mathpzc{K}^{\star}[B^{\dag}])^{\dag}\Sigma^{s}A^{\dag}]=\langle\hskip-.6mm\langle \mathpzc{K}^{\star}[B^{\dag}],A^{\dag}\rangle\hskip-.6mm\rangle_s=\langle\hskip-.6mm\langle B^{\dag},\mathpzc{K}[A^{\dag}]\rangle\hskip-.6mm\rangle_s=\mathrm{Tr}\big[\Sigma^{1-s}B\Sigma^{s}\,\mathpzc{K}[A^{\dag}]\big]. 
\end{align*}
Similarly,
\begin{align*}
\langle\hskip-.6mm\langle \mathpzc{R}_{\,s}\circ \mathpzc{K}[A],B\rangle\hskip-.6mm\rangle_{s}=\mathrm{Tr}\big[\Sigma^{1-s}\big(\Sigma^{1-2s}(\mathpzc{K}[A])\Sigma^{2s-1}\big)^{\dag}\Sigma^{s}B\big]= \mathrm{Tr}\big[\Sigma^{s}\mathpzc{K}[A^{\dag}]\Sigma^{1-s}B\big]= \mathrm{Tr}\big[\Sigma^{1-s}B\Sigma^{s}\mathpzc{K}[A^{\dag}]\big],
\end{align*}
\end{widetext}
 which coincides with the previous equation.

3. Let $\Sigma\ket{m}=h_{m}\ket{m}$, where $\Sigma$ defines the QDB condition \eqref{QDB1}. We have
\begin{equation}
\mathpzc{R}_{\,s}[\ket{m}\bra{n}]=(h_{n}/h_{m})^{2s-1}\ket{m}\bra{n},
\end{equation}
that is, the operators $|m\rangle\langle n|$ are the eigenoperators of $\mathpzc{R}_{\,s}$. These operators are orthogonal in the sense that $\langle\hskip-.6mm\langle |m\rangle\langle n|,|m'\rangle\langle n'|\rangle\hskip-.6mm\rangle_{s}\propto \delta_{mm'}\delta_{nn'}$, and the $\mathpzc{R}_{\,s}$-invariant subspace of $\mathpzc{H}_{\,\Sigma,s}$ corresponding to the eigenvalue $1$ is spanned by the eigenoperators $\{\ket{m}\bra{m}\}_{m=1}^{d}$. Since $\mathpzc{G}^{\sharp}_{\tau}$ preserves Hermiticity, part (2) of the lemma ensures that $\mathpzc{G}^{\sharp}_{\tau}\circ\mathpzc{R}_{\,s}=\mathpzc{R}_{\,s}\circ \mathpzc{G}^{\sharp}_{\tau}$ so that the invariant subspace of $\mathpzc{R}_{\,s}$ is mapped into itself by the dynamics as well as its orthogonal subspace linearly spanned by the eigenoperators $\ket{n}\bra{m}$ ($n\neq m$).
\end{proof}

\begin{theorem}
\label{TheoQDB1}
For a dynamical semigroup map $\mathpzc{G}^{\sharp}_{\tau}$, if its Lindblad generator $\mathpzc{L}^{\sharp}$ satisfies the first QDB condition \eqref{QDB1} with respect to the thermal state $\varrho^{(\beta_{\mathrm{f}})}$, then the QFR (\ref{FluctRel}) holds for all $\tau\geqslant 0$. 
\end{theorem} 

\begin{proof}
On the one hand, since $\mathpzc{L}^{\sharp}_{D}$ is self-adjoint on $\mathpzc{H}_{\,\Sigma,s}$, part (3) of the above lemma yields, for some $f_{m}\in\mathbb{R}$, 
\begin{equation}
\label{eq:adjoint-on-diagonal}
\mathpzc{L}^{\sharp}_{D}[\ket{n}\bra{n}]=\sum_{m=1}^{d} f_{m}\ket{m}\bra{m}.
\end{equation}
On the other hand, $\mathpzc{L}^{\sharp}_{H}[\ket{n}\bra{n}]=0$. Thus, using the Lie-Trotter relation \cite{book:Nielsen}
\begin{equation}
\mathpzc{G}^{\sharp}_{\tau}=e^{\tau(\mathpzc{L}^{\sharp}_{H} + \mathpzc{L}^{\sharp}_{D})}=\lim_{k\to\infty}\big(e^{(\tau/k) \mathpzc{L}^{\sharp}_H} e^{(\tau/k) \mathpzc{L}^{\sharp}_D}\big)^k,
\end{equation}
it follows that $\mathpzc{G}^{\sharp}_{\tau}[\ket{n}\bra{n}]=e^{\tau \mathpzc{L}^{\sharp}_{D}}[\ket{n}\bra{n}]$. Additionally, since $\mathpzc{L}^{\sharp}_{D}$ is a self-adjoint operator on $\mathpzc{H}_{\,\Sigma,s}$, such is $\mathpzc{G}^{\sharp}_{\tau}$; then part (1) of the above lemma---Eq. (\ref{eq:lemma-self-adjoint})---applies, 
\begin{equation*}
\bra{m} \mathpzc{G}^{\sharp}_{\tau}[\ket{n}\bra{n}]\ket{m}=e^{-\beta_{\mathrm{f}}(E_{n}-E_{m})}\bra{n}\,\mathpzc{G}^{\sharp}_{\tau}[\ket{m}\bra{m}]\ket{n},
\end{equation*}
whence Eq. \eqref{aux1} (and the QFR) holds.
\end{proof}

\subsection{QFR for general dynamical maps}
\label{subsec:fr-dgen}

Let us restrict ourselves to systems whose Hamiltonians are invariant under time-reversal, $\mathpzc{T}[H]=H$. It is then straightforward to see that this property carries over to all eigenprojectors $\ket{m}\bra{m}$ of the Hamiltonian $H$ too, 
\begin{equation}
\label{timerevaux0}
\mathpzc{T}[\ket{m}\bra{m}]=\ket{m}\bra{m}.
\end{equation}
This can be shown by noting $\mathpzc{T}[H^{k}]=(\mathpzc{T}[H])^{k}$ ($\forall k\in\mathbb{N}$) and employing the resolvent representation of the eigenprojectors as $\ket{m}\bra{m}=(2\pi i)^{-1} \oint_{c_m}(z\mathbbmss{I}-H)^{-1}\,\mathrm{d}z$, where $c_{m}$ is a circle centered around the eigenvalue $E_m$ of $H$ without encircling or passing over any other eigenvalue \cite{book:Hassani}. 

\begin{theorem}
\label{TheoQDB2}
If a dynamical map $\mathpzc{G}^{\sharp}_{\tau}$ satisfies the second QDB condition \eqref{QDB2} with respect to the thermal state $\varrho^{(\beta_{\mathrm{f}})}$ and the system Hamiltonian is invariant under time-reversal $\mathpzc{T}$, then the QFR (\ref{FluctRel}) holds for any time $\tau\geqslant 0$.
\end{theorem}

\begin{proof}
Setting $A=\ket{m}\bra{m}$ and $B=\ket{n}\bra{n}$ in Eq. \eqref{QDB2} with $H\ket{m}=E_m\ket{m}$, and using Eq. \eqref{timerevaux0}, one obtains
\begin{equation*}
\bra{m}\mathpzc{G}^{\sharp}_{\tau}[|n\rangle \langle n|]\ket{m}=e^{-\beta_{\mathrm{f}}(E_{n}-E_{m})}\bra{n}\mathpzc{G}^{\sharp}_{\tau}[\ket{m}\bra{m}]\ket{n},
\end{equation*}
which a case where Eq. (\ref{aux1}) applies.
\end{proof} 

A caveat is in order here. Although from Theorems \ref{TheoQDB1} and \ref{TheoQDB2} we see that when a thermalizing dynamics has the QDB property (in either forms), the QFR \eqref{FluctRel} holds, the converse is not necessarily valid. In fact, in the next section we present an example showing that the QDB and QFR are not equivalent to each other because one may have thermalization without the QDB condition.

\section{Examples}
\label{sec: Examples}

In the following, three examples are presented to highlight our results. The first example concerns a qubit dynamics which is thermalizing but not a semigroup with a Lindblad generator. In fact, this dynamics possesses an asymptotic thermal state which is not time-invariant (namely, not a fixed point of the dynamics). We show that the QFR does not hold; rather, a time-dependent correction appears in the ratio $R(\mathbbmss{E};\tau)$, which disappears asymptotically---in agreement with Theorem \ref{theo: infinity}. The second example is based on the so-called quantum optical master equation, which describes a two-level atom in interaction with the quantized electromagnetic field, the latter acting as a thermal environment at inverse temperature $\beta_{\mathrm{f}}$. We show that the resulting dynamics is FPT, thus Theorem \ref{theo: FP-thermalizing} applies in this case and the QFR holds. Moreover, such a dynamics respects the QDB condition, so that one can equivalently argue the validity of the QFR from the results of the previous section. The last example demonstrates that the QFR and QDB condition are not equivalent, providing an FPT semigroup dynamics which satisfies the QFR but fulfills neither QDB conditions.

\subsection{A thermalizing non-FPT dynamics}
\label{ex-1}

Here we consider an example of a \textit{non-Markovian} thermalizing map acting as a qubit ``generalized amplitude damping channel" \cite{book:Nielsen}. In particular, we show how by tuning some parameters of the dynamics in a suitable manner it is possible to construct a dynamics which is thermalizing but not FPT \cite{2017-Marcantoni}. In this case, it turns out that the QFR does not hold at finite times, whereas it is recovered asymptotically---as also expected from Theorem \ref{theo: infinity}. 

Consider a quantum operation (\ref{channel}) with
\begin{align*}
G_{\tau}^{(1)}&=\sqrt{q_{\tau}}\big(\ket{1}\bra{1}+\sqrt{1-\xi_{\tau}}\ket{2}\bra{2}\big),\\
G_{\tau}^{(2)}&=\sqrt{q_{\tau}\xi_{\tau}}\ket{1}\bra{2},\\
G_{\tau}^{(3)}&=\sqrt{1-q_{\tau}}\big(\sqrt{1-\xi_{\tau}}\ket{1}\bra{1}+\ket{2}\bra{2}\big),\\
G_{\tau}^{(4)}&=\sqrt{(1-q_{\tau})\xi_{\tau}}\ket{2}\bra{1},
\end{align*}  
where $q_{\tau},\xi_{\tau}\in[0,1]$. Given an initial state described by 
\begin{equation*}
\varrho(0)=d(0)\ket{1}\bra{1}+[1-d(0)] \ket{2}\bra{2}+k(0)\ket{1}\bra{2}\,+\,k^{*}(0)\ket{2}\bra{1},
\end{equation*} 
the state of the system at time $\tau$ becomes
\begin{align}
\varrho(\tau)=& d(\tau)\ket{1}\bra{1}+[1-d(\tau)] \ket{2}\bra{2}+k(\tau)\ket{1}\bra{2}\nonumber\\
& + k^{*}(\tau)\ket{2}\bra{1},
\end{align}
where
\begin{align}
d(\tau)&=(1-\xi_{\tau})d(0)+(1-q_{\tau})\xi_{\tau}\label{d-k-1}, \\
k(\tau)&=\sqrt{1-\xi_{\tau}}k(0).
\label{d-k-2}
\end{align}
Except for the constraint imposed by the initial condition, namely $\xi_{0}=0$, one can freely (but smoothly) adjust the parameters $q_{\tau}$ and $\xi_{\tau}$. As it is evident from Eqs. (\ref{d-k-1}) and (\ref{d-k-2}), one can impose a unique asymptotic state to exist for this dynamical map by means of the condition $\xi_{\infty}=1$; in other words, this condition guarantees that all the information related to the initial state $\varrho(0)$ is lost at long times. Moreover, requiring $q_{\infty}=(1/2)[1-\tanh(\beta_{\mathrm{f}}\omega/2)]$ ensures that the unique asymptotic state is indeed a thermal state at inverse temperature $\beta_{\mathrm{f}}$. Hence, such a dynamics is thermalizing but not FPT---unless we consider $q_{\tau}=q_\infty$ for any $\tau\geqslant 0$. In the following, however, we assume that $q_{\tau}$ is time-dependent.

From the transition probabilities
\begin{align}
p(\ket{1}\rightarrow\ket{2};\tau)&=\bra{2}\varrho(\tau|d(0)=1)\ket{2}=\xi_{\tau}q_{\tau},\\
p(\ket{2}\rightarrow\ket{1};\tau)&=\bra{1}\varrho(\tau|d(0)=0)\ket{1}=\xi_{\tau}(1-q_{\tau}),
\end{align}
we have
\begin{equation}
\frac{p(\ket{1}\rightarrow\ket{2};\tau)}{p(\ket{2}\rightarrow\ket{1};\tau)}=\frac{q_{\tau}}{1-q_{\tau}},
\end{equation}
which is independent of the parameter $\xi_{\tau}$, and in the limit $\tau\rightarrow\infty$ it tends to $e^{-\beta_{\mathrm{f}}(E_{2}-E_{1})}$. Introducing $f_{\tau}=q_{\infty}-q_{\tau}$ one obtains
\begin{equation}
R(\mathbbmss{E};\tau)=F(\tau)\,e^{\Delta\beta\,\mathbbmss{E}},
\label{Ft}
\end{equation}
where $F(\tau)=[1-f_{\tau}/q_{\infty}]/[1+f_{\tau}/(1-q_{\infty})]$. Note that $\lim_{\tau\to\infty}F(\tau)=1$, whence we retrieve $R(\mathbbmss{E};\infty)=e^{\Delta\beta\,\mathbbmss{E}}$---in agreement with Theorem \ref{theo: infinity}. Moreover, if $f_{\tau}=0\,\forall\tau$, the dynamics becomes FPT and---as expected---the QFR holds at finite times.

\subsection{An FPT dynamics satisfying the QDB condition}
\label{fpt-qbd}

In this example we consider a qubit evolving in time according to the so-called ``quantum optical master equation" \cite{2002-Breuer}
\begin{align}
\partial_{\tau} \varrho=& -i\frac{\omega}{2} \left[ \sigma_{z} , \varrho \right] +  \gamma \bar{n}(\omega,\beta) \big( \sigma_{+} \varrho \sigma_{-} - \frac{1}{2}\big\lbrace \sigma_{-} \sigma_{+}, \varrho \big\rbrace \big)\nonumber\\
&+ \gamma \big(\bar{n}(\omega,\beta) + 1\big)\big( \sigma_{-} \varrho\sigma_{+} - \frac{1}{2}\big\lbrace \sigma_{+} \sigma_{-}, \varrho \big\rbrace \big), \label{LindOpt}
\end{align}
where $\sigma_{\pm}= (1/2)(\sigma_{x} \pm i\sigma_{y})$, $\gamma$ is a positive damping rate, $\bar{n}(\omega,\beta_{\mathrm{f}}) = (e^{\beta_{\mathrm{f}} \omega }-1)^{-1}$ is the bosonic occupation number in thermal equilibrium, and we have dropped the $\tau$-dependence of $\varrho$ to lighten the notation. Note also that here we have used the convention $\sigma_{z}|i\rangle=(-1)^{i}|i\rangle$ ($i\in\{1,2\}$) in this example. This kind of dynamics is used, for example, to model a two-level atom interacting with a thermal bath of photons at inverse temperature $\beta_{\mathrm{f}}$ and has been widely studied in the literature. Here we investigate the validity of the QFR \eqref{FluctRel} and the QDB condition (\ref{QDB1}) for this dynamics.

First, we show that this dynamics is FPT, and hence the QFR holds because of Theorem \ref{theo: FP-thermalizing}. If we parametrize the state in the Bloch form as 
\begin{equation}
\varrho(\tau)=(1/2)\big[\mathbbmss{I}+\mathbf{r}(\tau)\cdot\bm{\sigma}\big],
\label{Bloch}
\end{equation}
where $\mathbf{r}=(r_{x},r_{y},r_{z})$ is a vector with $\Vert \mathbf{r}\Vert\leqslant 1$ and $\bm{\sigma}=(\sigma_{x},\sigma_{y},\sigma_{z})$, the solution of the dynamics is given by 
\begin{widetext}
\begin{align}
\varrho(\tau)= \frac{1}{2} \begin{pmatrix}1+r_{z}(0) e^{-\bar{\gamma}\tau}+ \tanh(\beta_{\mathrm{f}} \omega /2)[e^{-\bar{\gamma}\tau}-1] & [r_{x}(0)-i r_{y}(0)] e^{-(i\omega+\bar{\gamma}/2)\tau}\\ [r_{x}(0)+i r_{y}(0)] e^{-(-i\omega + \bar{\gamma}/2)\tau} & 1-r_{z}(0) e^{-\bar{\gamma}\tau}-\tanh(\beta_{\mathrm{f}} \omega/2) [e^{-\bar{\gamma}\tau}-1]
\end{pmatrix},
\end{align}
\end{widetext}
where we have set $\bar{\gamma}= \gamma \big[2 \bar{n}(\omega,\beta_{\mathrm{f}}) + 1\big]= \gamma \coth(\beta_{\mathrm{f}} \omega /2)$. From this solution it is immediate to see that this is a thermalizing dynamics, for any initial condition the system relaxes to a thermal state $\varrho(\infty)=\varrho^{(\beta_{\mathrm{f}})}$, corresponding to the Hamiltonian $H=(\omega/2)\sigma_{z}$.

In addition, note that the thermal state $\varrho^{(\beta_{\mathrm{f}})} $ is also stationary because the dynamics \eqref{LindOpt} obeys the semigroup composition law. That is, this dynamics is FPT, and hence the QFR is met due to Theorem \ref{theo: FP-thermalizing}. 

Alternatively, one could also have argued that the QFR is met because of Theorem \ref{TheoQDB1}. Indeed, according to Ref. \cite{2007-Fagnola}, in the qubit case, the most general semigroup of CPTP maps satisfying the QDB condition \eqref{QDB1} with respect to the state $\varrho^{(\beta_{\mathrm{f}})}$ is the solution of the following Lindblad master equation:
\begin{align}
\partial_{\tau} \varrho=& -i \frac{\omega}{2} \left[ \sigma_{z} , \varrho \right] + \mu e^{\beta_{\mathrm{f}} \omega}\big( \sigma_{-} \varrho \sigma_{+} - \frac{1}{2}\big\lbrace \sigma_{+} \sigma_{-}, \varrho \big\rbrace \big)\nonumber\\
&+ \mu \big( \sigma_{+} \varrho \sigma_{-} - \frac{1}{2}\big\lbrace \sigma_{-} \sigma_{+}, \varrho \big\rbrace \big) + \eta \big( \sigma_{z} \varrho \sigma_{z} - \varrho\big),
\label{LindQDB}
\end{align}
where $\mu$ and $\eta$ are positive parameters. One can thus clearly observe that Eq. \eqref{LindOpt} is a special case of the latter when $\eta=0$ and $\mu= \gamma \bar{n}(\omega,\beta_{\mathrm{f}})$---note that $\bar{n}(\omega,\beta_{\mathrm{f}})+1= \bar{n}(\omega,\beta_{\mathrm{f}}) e^{\beta_{\mathrm{f}} \omega}$. Thus, the quantum optical master equation describes a dynamics satisfying both the QDB condition and the QFR. 

The following example, however, provides an instance of dynamics where the QFR \eqref{FluctRel} is fulfilled but neither QDB conditions \eqref{QDB1} and \eqref{QDB2} hold.

\subsection{An FPT dynamics not satisfying the QDB condition}
\label{sec;fpt=!qdb}

This example demonstrates that the QDB condition is not equivalent to the QFR \eqref{FluctRel}. Indeed, in the following we present a dynamics for a qubit which is FPT---thus obeying the QFR according to Theorem \ref{theo: FP-thermalizing}---but does not satisfy the QDB condition. In doing so, it is more convenient to vectorize the state of the system as $| \varrho \rangle \equiv (1,r_{x},r_{y},r_{z})$. Any linear operation acting on $\varrho$ can then be represented as a $4 \times 4$ matrix acting on the vector $| \varrho \rangle$.

As already mentioned in the first example (\ref{ex-1}), according to Ref. \cite{2007-Fagnola}, in the qubit case, the most general semigroup of CPTP maps satisfying the QDB condition \eqref{QDB1} with respect to the state $\varrho^{(\beta_{\mathrm{f}})}$ is the solution of the Lindblad master equation \eqref{LindQDB}. This equation can be recast as
\begin{equation}
\partial_{\tau}| \varrho(\tau) \rangle =-2 \mathbbmss{L} | \varrho(\tau) \rangle,
\label{masterQDB}
\end{equation}
where
\begin{equation}
\label{aux1a}
\mathbbmss{L}= \begin{pmatrix}
0 & 0 & 0 & 0       \\
0 & \odot & \omega/2  & 0 \\
0 & -\omega/2 & \odot & 0 \\
\odot_{-} & 0 & 0 & \odot_{+}
\end{pmatrix},
\end{equation}
with 
\begin{align*}
\odot =& \eta+\mu(1+e^{\beta_{\mathrm{f}} \omega}),\\
\odot_{\pm} =& 2\mu(1\pm e^{\beta_{\mathrm{f}} \omega}).
\end{align*}

Now consider another qubit dynamics which is a generalized form of the above one and is described by
\begin{equation}
\label{masterTherm}
\mathbbmss{L}_{\mathrm{th}}= \begin{pmatrix}
0 & 0 & 0 & 0       \\
0 & \nu & \omega/2  & 0 \\
0 & -\omega/2 & \alpha & 0 \\
\chi & 0 & 0 & \zeta
\end{pmatrix}.
\end{equation}
This dynamics is physically legitimate because it is CPTP \cite{BFRev}. Moreover, it can be solved analytically, 
\begin{align}
r_{i}(\tau) &= u_{i+} e^{\tau k_+}+ u_{i-} e^{\tau k_-},\, i\in\{x,y\},\nonumber\\
r_{z}(\tau)&= e^{-2\zeta \tau}r_{z}(0)- \big(1-e^{-2\zeta \tau}\big) \chi/\zeta, \nonumber
\end{align}
where 
\begin{align}
u_{x\pm} &= \pm \frac{k_{\mp}+2\nu}{k_{-} \pm k_{+}} r_{x}(0) \pm \frac{\omega}{k_{-}\pm k_{+}} r_{y}(0), \nonumber \\
u_{y\pm} &= \pm \frac{k_{\mp}+2\alpha}{k_{-} \pm k_{+}} r_{y}(0) \mp \frac{\omega}{k_{-}\pm k_{+}} r_{x}(0), \nonumber \\
k_{\pm}&= -(\alpha +\nu) \pm i\sqrt{\omega^{2}-(\alpha-\nu)^{2}}. \nonumber
\end{align}
This solution implies that the dynamics is FPT. There is a unique asymptotic state $| \varrho(\infty) \rangle = (1,0,0,-\chi/\zeta)$ which is also a fixed point at finite times---because the dynamics obeys the semigroup composition law. Moreover, by comparison we see that the time evolution satisfying the QDB condition \eqref{masterQDB} is a special case of Eq. \eqref{masterTherm} in which 
\begin{align*}
\chi =& 2\mu(1-e^{\beta_{\mathrm{f}} \omega}), \\
\zeta =& 2\mu(1+e^{\beta_{\mathrm{f}} \omega}), \\
\nu=&\alpha= \eta+\mu(1+e^{\beta_{\mathrm{f}} \omega}).
\end{align*}

Since in general one can have $\nu\neq \alpha$, it can be concluded that there exist cases of qubit FPT dynamics which do not satisfy the QDB condition \eqref{QDB1}. Moreover, the time-invariant asymptotic state corresponds to a thermal state $\varrho^{(\beta_{\mathrm{f}})}$ with $H\propto \sigma_{z}$. Thus, the QDB condition \eqref{QDB2} with time-reversal implemented by complex conjugation relative to the $\sigma_{z}$ eigenbasis ($\mathpzc{C}$ as in appendix \ref{app:Theta}) is not satisfied. As a result then the QFR \eqref{FluctRel} is not equivalent in general to having the QDB condition in either forms (\ref{QDB1}) and (\ref{QDB2}).

\section{Summary}
\label{sec:conc}

We have considered thermalizing open quantum dynamics in its general form and also the special case \blue{of} dynamics with semigroup property generated by Lindblad generators. We have formulated an extended version of the quantum fluctuation relation which compares the probabilities of absorbing a given amount of energy and releasing the same amount to the environment---both evaluated for the forward-time dynamics. We then have sought sufficient conditions for an open-system dynamics to fulfill the quantum fluctuation relation. Specifically, we have shown that: (i) all thermalizing dynamics (irrespectively of the Hilbert-space dimension of the system) satisfy the quantum fluctuation relation asymptotically; (ii) if the thermalizing dynamics for qubits has the property that its thermal state is also stationary, it shows the quantum fluctuation property at any finite time too; and (iii) the quantum detailed balance condition (for each type of open dynamics) suffices to satisfy the proposed quantum fluctuation relation. We have, however, argued that the quantum fluctuation relation and the quantum detailed balance condition are not equivalent; one can find the former without the latter. 

Our study may shed light on how two important concepts in quantum thermalization and second law of thermodynamics for quantum systems are connected, and may enable further analysis of the emergence of other peculiar features of thermalizing dynamics or when systems may thermalize.

\textit{Acknowledgements.}---Initial inputs of S. Alipour are acknowledged. This work was partially supported by Iran's National Elites Foundation (to M.R.), the Project EU 2020 ERC-2015-STG G.A. No. 677488 (to S.M.), and Sharif University of Technology's Office of Vice President for Research (to A.T.R.)


\appendix
\section{Review of the time-reversal operations $\Theta$ and $\mathpzc{T}$}
\label{app:Theta}

In standard quantum mechanics \cite{book:Sakurai}, it is argued that the time-reversal operator $\Theta:\mathpzc{H}\mapsto \mathpzc{H}$ is an \textit{antiunitary} operator defined by its action on the position and momentum operators, 
\begin{align}
\Theta\bm{R}\Theta^{\dag} &= \bm{R},\\
\Theta\bm{P}\Theta^{\dag}&= -\bm{P},
\end{align}
where $\bm{R}=(x,y,z)$ and $\bm{P}=(p_{x},p_{y},p_{z})$. By antiunitarity we mean the following properties satisfied together:
\begin{align}
\Theta(\alpha_{1}|v_{1}\rangle + \alpha_{2}|v_{2}\rangle)&= \alpha^{*}_{1}\Theta(|v_{1}\rangle) + \alpha^{*}_{2}\Theta(|v_{2}\rangle),\\
\langle \widetilde{v}_{2}|\widetilde{v}_{1}\rangle &= \langle v_{1}|v_{2}\rangle,
\label{tt}
\end{align}
for all $\alpha_{1},\alpha_{2}\in\mathbb{C}$ and $|v_{1}\rangle,|v_{2}\rangle\in\mathpzc{H}$, with $|\widetilde{v}\rangle\equiv \Theta(|v\rangle)$. Equation (\ref{tt}) implies that $\Theta^{\dag}\Theta=\Theta\Theta^{\dag}=\mathbbmss{I}$. As a result, if $\{|e_{j}\rangle\}_{j=1}^{d}$ is an orthonormal basis set for $\mathpzc{H}$, so is the set $\{\Theta(|e_{j}\rangle)\}_{j=1}^{d}$. More importantly, one can show that for all $A\in\ M_{d}(\mathbb{C})$ and $|v_{1}\rangle,|v_{2}\rangle\in\mathpzc{H}$ \cite{book:Sakurai}
\begin{equation}
\langle v_{1}|A|v_{2}\rangle = \langle \widetilde{v}_{2}|\Theta\big( A^{\dag}\Theta^{\dag}(|\widetilde{v}_{1}\rangle)\big).
\label{TT}
\end{equation}

For \textit{spinless} quantum systems, it can be seen that $\Theta$ is tantamount to complex conjugation $\mathpzc{C}$ with respect to a chosen orthonormal basis $\{|e_{j}\rangle\}_{j=1}^{d}$ for $\mathpzc{H}$, which is the antilinear operation defined through 
\begin{equation}
\mathpzc{C}[|\psi\rangle] = \mathpzc{C}\Big[\sum_{j}\langle e_{j}|\psi\rangle\vert e_{j}\rangle\Big] = \sum_{j} \langle e_{j} |\psi\rangle^{*} |e_j\rangle.
\end{equation}
Note that $\mathpzc{C}^{\dag}=\mathpzc{C}$ and $\mathpzc{C}\mathpzc{C}^{\dag}=\mathpzc{C}^{\dag}\mathpzc{C}=\mathbbmss{I}$. Obviously, $\mathpzc{C}$ in the position representation we have $\mathpzc{C}^{\dag}\bm{R}\mathpzc{C}=\bm{R}$; whereas $\mathpzc{C}^{\dag}\bm{P}\mathpzc{C}=-\bm{P}$, from whence the orbital angular momentum operator ($\bm{L}=\bm{R}\times \bm{P}$) fulfills $\mathpzc{C}^{\dag}\bm{L}\mathpzc{C} =-\bm{L}$. Similarly, for two-level systems, $\mathpzc{C}$ with respect to the $\sigma_{z}$-eigenbasis acts as $\mathpzc{C}^{\dag} \sigma_{x} \mathpzc{C}=\sigma_{x}$, $\mathpzc{C}^{\dag}\sigma_{y} \mathpzc{C}=-\sigma_{y}$, and $\mathpzc{C}^{\dag}\sigma_{z} \mathpzc{C}=\sigma_{z}$.

The situation is, however, different for \textit{general} angular momentum $\mathbf{J}$ (including \textit{spin} degree of freedom; $\bm{J}= \bm{L} + \bm{S}$). Although $\mathpzc{C}^{\dag}\bm{J}\mathpzc{C}=-\bm{J}$, it is shown that in general 
\begin{equation}
\Theta=e^{-i\pi J_{y}}\mathpzc{C},
\end{equation}
which for spin-$1/2$ reduces to $\Theta =-i\sigma_{y}\mathpzc{C}$. As a result, $\Theta^{2}=+\mathbbmss{I}$ for integer and $\Theta^{2}=-\mathbbmss{I}$ for half-integer $\bm{J}$---see Ref. \cite{book:Sakurai} for a detailed discussion.

Equation (\ref{TT}) motivates the definition of the time-reversal operation $\mathpzc{T}$ as in Eq. (\ref{Trevaut}). One can show that \cite{thesis} $\mathpzc{T}$ is indeed a \textit{linear}, norm- and trace-preserving map in the sense that: (i) $\mathpzc{T}[\alpha_{1}A_{1} + \alpha_{2}A_{2}] = \alpha_{1}\mathpzc{T}[A_{1}] + \alpha_{2}\mathpzc{T}[A_{2}]$ ($\forall\alpha_{1},\alpha_{2}\in\mathbb{C}$ and $\forall A_{1},A_{2}\in M_{d}(\mathbb{C})$); (ii) $\Vert \mathpzc{T}[A]\Vert =\Vert A\Vert$ ($\forall A\in M_{d}(\mathbb{C})$), where $\Vert \cdot\Vert$ is the induced norm on $M_{d}(\mathbb{C})$ \cite{book:Hassani}; and (iii) $\mathrm{Tr}\big[\mathpzc{T}[A]\big]=\mathrm{Tr}[A]$ ($\forall A\in M_{d}(\mathbb{C})$); which has the extra properties (iv) $\mathpzc{T}[A^{\dag}]=(\mathpzc{T}[A])^{\dag}$ ($\forall A\in M_{d}(\mathbb{C})$), (v) $\mathpzc{T}[AB]=\mathpzc{T}[B]\,\mathpzc{T}[A]$ ($\forall A,B\in M_{d}(\mathbb{C})$), (vi) $\mathpzc{T}[A]$ being a \textit{linear} operator $\in M_{d}(\mathbb{C})$ for all $A\in M_{d}(\mathbb{C})$, and (vii) $\mathpzc{T}\circ \mathpzc{T}=\mathpzc{I}$ (with $\mathpzc{I}$ being the identity map).



\begin{thebibliography}{100}
\bibitem{hbar}We have assumed $\hbar\equiv 1$ throughout this paper.
\end{thebibliography}
\end{document}